\documentclass{amsart}
\usepackage{breakurl,amsmath,amsthm,amssymb,xspace,xypic,enumerate,color,tikzfig}

\tikzstyle{every picture}=[baseline=-0.25em,shorten <=-0.1pt]
\tikzstyle{dotpic}=[scale=0.5]
\tikzstyle{braceedge}=[decorate,decoration={brace,amplitude=1mm,raise=-1mm}]
\tikzstyle{dot}=[inner sep=0.7mm,minimum width=0pt,minimum height=0pt,fill=black,draw=black,shape=circle]
\tikzstyle{black dot}=[dot]
\tikzstyle{white dot}=[dot,fill=white]
\tikzstyle{gray dot}=[dot,fill=gray!40!white]
\tikzstyle{alt white dot}=[white dot,label={[xshift=3mm,yshift=-0.05mm,font=\tiny]left:$*$}]
\tikzstyle{alt gray dot}=[gray dot,label={[xshift=3mm,yshift=-0.05mm,font=\tiny]left:$*$}]
\tikzstyle{white norm}=[rectangle,fill=white,draw=black,minimum height=2mm,minimum width=2mm,inner sep=0pt,font=\small]
\tikzstyle{gray norm}=[white norm,fill=gray!40!white]
\tikzstyle{square box}=[rectangle,fill=white,draw=black,minimum height=5mm,minimum width=5mm,font=\small]
\tikzstyle{square gray box}=[rectangle,fill=gray!30,draw=black,minimum height=6mm,minimum width=6mm]
\tikzstyle{diredge}=[->]
\tikzstyle{cdedge}=[-latex]
\tikzstyle{rdiredge}=[<-]
\tikzstyle{dashed edge}=[dashed]
\tikzstyle{cross}=[preaction={draw=white, -, line width=3pt}]
\tikzstyle{point}=[regular polygon,regular polygon rotate=180,regular polygon sides=3,draw=black,scale=0.75,inner sep=-0.5pt,minimum width=.75cm,fill=white]
\tikzstyle{copoint}=[regular polygon,regular polygon sides=3,draw=black,scale=0.75,inner sep=-0.5pt,minimum width=.75cm,fill=white]

\newcommand{\pantsalg}{%
\,\begin{tikzpicture}[dotpic,scale=0.8]
		\node [style=none] (0) at (-0.75, -0.5) {};
		\node [style=none] (1) at (0.25, 0.5) {};
		\node [style=none] (2) at (0.75, -0.5) {};
		\node [style=none] (3) at (0.25, -0.5) {};
		\node [style=none] (4) at (-0.25, 0.5) {};
		\node [style=none] (5) at (-0.25, -0.5) {};
		\draw [style=diredge, in=-90, out=90] (2.center) to (1.center);
		\draw [style=diredge, in=90, out=-90] (4.center) to (0.center);
		\draw [style=diredge, in=90, out=90, looseness=2.00] (5.center) to (3.center);
\end{tikzpicture}\,}
\newcommand{\dotcounit}[1]{%
\,\begin{tikzpicture}[dotpic,yshift=-1mm]
\node [#1] (a) at (0,0.25) {}; 
\draw [diredge] (0,-0.2)--(a);
\end{tikzpicture}\,}
\newcommand{\dotunit}[1]{%
\,\begin{tikzpicture}[dotpic,yshift=1.5mm]
\node [#1] (a) at (0,-0.25) {}; 
\draw [diredge] (a)--(0,0.2);
\end{tikzpicture}\,}
\newcommand{\dotcomult}[1]{%
\,\begin{tikzpicture}[dotpic,yshift=0.5mm]
	\node [#1] (a) {};
	\draw [diredge] (-90:0.55)--(a);
	\draw [diredge] (a) -- (45:0.6);
	\draw [diredge] (a) -- (135:0.6);
\end{tikzpicture}\,}
\newcommand{\dotmult}[1]{%
\,\begin{tikzpicture}[dotpic]
	\node [#1] (a) {};
	\draw [diredge] (a) -- (90:0.55);
	\draw [diredge] (a) (-45:0.6) -- (a);
	\draw [diredge] (a) (-135:0.6) -- (a);
\end{tikzpicture}\,}

\newcommand{\dotdualmult}[1]{%
\!\begin{tikzpicture}[dotpic]
		\node [style=white dot] (0) at (0, 0.3) {};
		\node [style=none] (1) at (-0.5, -0.4) {};
		\node [style=none] (2) at (0.5, -0.4) {};
		\node [style=none] (3) at (0, 0.8) {};
		\draw [style=diredge] (3.center) to (0);
		\draw [style=diredge, in=15, out=-30, looseness=1.50] (0) to (1.center);
		\draw [style=diredge, in=165, out=-150, looseness=1.50] (0) to (2.center);
\end{tikzpicture}\!}

\newcommand{\dotnorm}[1]{%
\,\begin{tikzpicture}[dotpic,yshift=0.4mm]
		\node [style=none] (0) at (0, -0.4) {};
		\node [style=white norm] (1) at (0, -0) {};
		\node [style=none] (2) at (0, 0.5) {};
		\draw (0.center) to (1);
		\draw [style=diredge] (1) to (2.center);
\end{tikzpicture}\,}

\newcommand{\whiteunit}{\dotunit{white dot}}
\newcommand{\whitecounit}{\dotcounit{white dot}}
\newcommand{\whitemult}{\dotmult{white dot}}
\newcommand{\whitecomult}{\dotcomult{white dot}}

\newcommand{\graymult}{\dotmult{gray dot}}

\newcommand{\prodmult}[2]{\dotmult{#1}\!\!\!\!\dotmult{#2}}

\newcommand{\cat}[1]{\ensuremath{\mathbf{#1}}\xspace}
\newcommand{\FHilb}{\cat{FHilb}}
\newcommand{\Rel}{\cat{Rel}}

\DeclareMathOperator{\Mor}{Mor}
\DeclareMathOperator{\Proj}{Proj}
\DeclareMathOperator{\dom}{dom}

\theoremstyle{plain}
\newtheorem{theorem}{Theorem}[section]
\newtheorem{lemma}[theorem]{Lemma}
\newtheorem{proposition}[theorem]{Proposition}
\newtheorem{corollary}[theorem]{Corollary}
\theoremstyle{definition}
\newtheorem{definition}[theorem]{Definition}
\newtheorem{example}[theorem]{Example}
\newtheorem{remark}[theorem]{Remark}

\title{Compositional Quantum Logic}
\author{Bob Coecke}
\author{Chris Heunen}
\author{Aleks Kissinger}
\address{Department of Computer Science, University of Oxford, Wolfson
  Building, Parks Road, OX1 3QD, Oxford, UK}
\email{\{coecke,heunen,alek\}@cs.ox.ac.uk}
\date{\today}

\begin{document}
\begin{abstract}
  Quantum logic aims to capture essential quantum mechanical structure
  in order-theoretic terms. 
  The Achilles' heel of quantum logic is the absence of a canonical
  description of composite systems, given descriptions of their
  components.  
  We introduce a framework in which order-theoretic structure
  comes with a primitive composition operation.  The order is
  extracted from a generalisation of C*-algebra that applies to arbitrary dagger symmetric monoidal categories, which also provide the composition operation.
  In fact, our construction is entirely compositional,  without any additional assumptions on limits or enrichment.  Interpreted in the category of finite-dimensional Hilbert spaces,   it yields the projection lattices of arbitrary finite-dimensional  C*-algebras. Interestingly, there are  models that falsify
  standardly assumed correspondences, most notably the correspondence between noncommutativity of the algebra and nondistributivity of the order.
\end{abstract}
\maketitle

\section{Introduction}

In 1936, Birkhoff and von Neumann questioned
whether the full Hilbert space structure was needed to capture the
essence of quantum mechanics~\cite{BvN}.  They argued that the
order-theoretic structure of the closed subspaces of state space, or
equivalently, of the projections of the operator algebra of
observables, may already tell the entire story. To be more precise, we need to consider an order
together with an order-reversing involution on it, a so-called
\emph{orthocomplementation}, which can also be cast as an
\emph{orthogonality relation}.  Support along those lines comes from
Gleason's theorem~\cite{Gleason}, which characterises the Born rule in terms
of order-theoretic structure. In turn, via Wigner's
theorem~\cite{Wigner}, this fixes unitarity of the dynamics.   

These developments prompted Mackey to formulate his programme for
the mathematical foundations of quantum mechanics: the reconstruction of
Hilbert space from  operationally meaningful  axioms on  an order-theoretic
structure~\cite{Mackey}.   In 1964, Piron ``almost'' completed that
programme for the infinite-dimensional case~\cite{Piron64,Piron}. Full completion was achieved much more recently, by Sol\`er in 1995~\cite{soler:orthomodular}.\footnote{See also the survey~\cite{StubbeSteirt}, which provides a comprehensive overview of the entire reconstruction, drawing from the fundamental theorem of  projective geometry. Reconstructions of quantum theory have recently seen a great revival~\cite{HardyAxiom, Chiri2}. In contrast to the Piron-Sol\`er theorem, this more recent work is mainly restricted to the finite-dimensional case, and focuses on operational axioms concerning how (multiple) quantum and classical systems interact.}

Birkhoff and von Neumann coined the term `quantum logic', in
light of the developments in algebraic logic which were also
subject to an order-theoretic paradigm. In particular they observed
that the distributive law for meets and joins, which is key to the
deduction theorem in classical logic, fails to hold for the lattice of
closed subspaces for a Hilbert space~\cite{BvN}.   

This failure of distributivity and hence the absence of a deduction
theorem resulted in rejection of the quantum `logic' idea by a
majority of logicians. However, while the name quantum logic was
retained, many of its researchers also rejected the direct link to
logic, and simply saw quantum logic as the study of the
order-theoretic structure associated to quantum phenomena, as well as
other structural paradigms that were proposed thereafter~\cite{Foulis,
  Ludwig}. 


\subsection*{The quantum logic paradigm.}   In the Mackey-Piron-Sol\`er reconstruction, the elements of the partially ordered set become
the projections on the resulting Hilbert space, that is, the self-adjoint
idempotents of the algebra of operators on the Hilbert space:
\begin{equation}
  p \circ p = p,  \qquad \qquad p^\dag= p.
\end{equation}
Conversely, the ordering  can be recovered from the composition structure on these projections:
\begin{equation}
  p \leq q \;\;\Longleftrightarrow\;\; p\circ q= p,
\end{equation}
and the orthogonality relation can be recovered from it, too: 
\begin{equation}
  p \perp q \;\;\Longleftrightarrow\;\; p \circ q = 0.
\end{equation}

In fact, the reconstruction does not  
produce Hilbert space, but Hilbert
space with superselection rules. That is, depending on the particular
nature of the ordering that we start with, it either 
produces quantum theory or classical theory, or combinations thereof. 

The presence of ``quantumness'' is famously heralded in order-theoretic
terms by the failure of the distributive law, giving rise to the
following comparison. 
\[
  {\mbox{classical} \over \mbox{quantum}}
  \;\;\simeq\;\; {\mbox{distributive} \over \mbox{nondistributive}}
\]
This translates as follows to the level of operator algebra.
\[
  {\mbox{classical} \over \mbox{quantum}}
  \;\;\simeq\;\;
  {\mbox{commutative} \over \mbox{noncommutative}}
\]
Thus, the combination yields the following slogan.
\begin{equation}\label{eq:NonComDist}
  {\mbox{distributive} \over \mbox{nondistributive}}
  \;\;\simeq {\mbox{commutative}\;\; \over \mbox{noncommutative}}
\end{equation}
This is indeed the case for the projection lattices of arbitrary
von Neumann algebras: the projection lattice is distributive if and only if the 
algebra is commutative~\cite[Proposition~4.16]{redei:quantumlogic},
and has been a guiding thought within the quantum structures
research community. 

\subsection*{Categorical quantum mechanics.} More recently, drawing on
modern developments in logic and computer science, and mainly a branch
called type-theory, Abramsky and Coecke introduced a radically
different approach to quantum structures that has gained prominence,
which takes \emph{compositional structure} as the starting
point~\cite{AC1}.  Proof-of-concept was provided by the fact that many 
quantum information protocols which crucially rely on the description
of compound quantum systems could be very succinctly derived at a high
level of abstraction. 

In what is now known as categorical quantum mechanics, composition of
systems is treated as a primitive connective, typically as a so-called
\emph{dagger symmetric monoidal category}. Additional axioms
may then be imposed on such categories to capture the particular
nature of quantum compoundness. In other words, a set of
equations that axiomatise the Hilbert space tensor products is
generalised to a broad range of theories.  Importantly, at no point is
an underlying vector-space like structure assumed.    

In contrast to quantum logic, this approach led to an abstract
language with high expressive power, that enabled one to address concretely
posed problems in the area of quantum computing
(see e.g.~\cite{BoixoHeunen, CK, DP2, Clare}) and quantum foundations
(see e.g.~\cite{CES}), and that has even led to interesting connections between
quantum structures and the structure of natural language~\cite{CSC, QLog}.    

One of the key insights of this approach is the fact that many notions
that are primitive in Hilbert space theory, and hence quantum theory,
can actually be recovered in compositional terms.  For example,
given the pure operations of a theory, one can define  mixed
operations in purely compositional terms, which together give rise to
a new dagger symmetric monoidal category~\cite{SelingerCPM}.    We
will refer to this construction, as (Selinger's)
\emph{CPM--construction}.  While this construction applies to  arbitrary dagger
symmetric monoidal categories (as shown in~\cite{SelingerAxiom, 
  CoeckeHeunen}), Selinger also assumed
\emph{compactness}~\cite{KellyLaplaza}, something that we will also do
in this paper.  These structures are called \emph{dagger compact categories}. 

Another example, also crucial to this paper, is the fact that
orthonormal bases can be expressed purely in terms of certain
so-called \emph{dagger Frobenius algebras}, which only rely on dagger
symmetric monoidal structure~\cite{CPV,abramskyheunen:hstar}.     
In turn, 
these dagger Frobenius algebras enable one to define derived concepts
such as stochastic maps. All of this still occurs within the language
of dagger symmetric monoidal categories~\cite{CPaqPav}.  We will refer
to this construction as the \emph{Stoch--construction}.
Similarly, 
finite-dimensional C*-algebras can also be realised as certain dagger Frobenius algebras, 
internal to the dagger compact category of finite-dimensional Hilbert spaces and
linear maps, the tensor product, and the linear algebraic adjoint~\cite{VicaryCstar}.   

Recently~\cite{QPL_CPstar}, the authors proposed a construction, called the
\emph{CP*--construction}, that generalises this correspondence to
certain dagger Frobenius algebras in arbitrary dagger compact
caterories. At the same time, this construction unifies the CPM--construction 
and the Stoch-construction, starting from a given dagger
compact category.  The resulting structure is an abstract approach to
classical-quantum interaction, with Selinger's CPM--fragment playing 
the role of the ``purely quantum'', and the abstract stochastic maps fragment
playing the role of the ``purely classical''.\footnote{There is an earlier
  unification of the CPM--construction and the
  Stoch-construction~\cite{SelingerIdempotent}, into which our construction faithfully
  embeds, see~\cite{QPL_CPstar}.  However, 
  this construction does not support the interpretation of
  ``generalised C*-algebras''~\cite{QPL_CPstar}.}   

\subsection*{Overview of this paper.} In this paper, we take this
framework of ``generalised C*-algebras'' as a starting point, and
investigate the structure of the dagger idempotents.  
We will refer to these as in short as projections too, since these
dagger idempotents provide the abstract counterpart to projections of
concrete C*-algebras.  

We show that, just as in the concrete case, one
always obtains a partially ordered set with an orthogonality relation.
However, equation~\eqref{eq:NonComDist} breaks down in general.
More specifically, in the dagger compact category  of sets and relations with the Cartesian
product as tensor  and the relational converse as the dagger, there
are commutative algebras with  nondistributive projection lattices.

As mentioned above, the upshot of our approach is that it resolves a
problem that rendered quantum logic useless for modern purposes: 
providing an order structure representing compound systems at an abstract level, given 
the  ones describing the component systems. Since we start
with a category with monoidal structure, of course composition for
objects is built in from the start, and it canonically lifts to algebras thereon.  
Let us emphasise that our framework relies solely on dagger
categorical and compositional structure: the (sequential) composition of
morphisms, and the (parallel) tensor product of morphisms. This is a key
improvement over previous work~\cite{harding:link,
  heunenjacobs:daggerkernellogic, harding:baer, jacobs:baer}  that combines
order-theoretic and compositional structure.\footnote{The
construction in~\cite{harding:link} needs the rather strong extra assumption of dagger
biproducts, while the construction
in~\cite{heunenjacobs:daggerkernellogic} requires the weaker
assumption of dagger kernels. The intersection of both constructions
can be made to work, provided one additionally assumes a weak form of additive enrichment~\cite{harding:baer}.}
 

\section{Background}

For background on symmetric monoidal categories we refer to the
existing literature on the subject
~\cite{CatsII}.  In particular we will
rely on their graphical representations, which are surveyed 
in~\cite{SelingerSurvey}.   

Diagrams will be read from bottom to top. Wires represent the
objects of the category, while boxes or dots or any other entity with
incoming and outcoming wires -- possibly none -- represents a
morphism, and their type is determined by the respective number of
incoming and outgoing wires.  The directions of arrows on wires
represent \emph{duals} of the compact structure.  

Our main objects of study are symmetric Frobenius algebras, defined as follows. Let us emphasise that this is a larger class of Frobenius algebra than just the commutative ones, which previous works on categorical quantum mechanics have mainly considered.

\begin{definition}
 Let $(\cat{C}, \otimes, I)$ be a symmetric monoidal category which  carries a \em dagger structure\em, that is, an identity-on-objects contravariant involutive endofunctor $\dagger: \cat{C}^{op}\to \cat{C}$. 
A \emph{Frobenius algebra}  in $\cat{C}$ is an object $A$ of $\cat{C}$  together  with morphisms 
\[
\whitemult: A\otimes A\to A
\qquad
\whiteunit: I\to A
\qquad
\whiteunit: A\to A\otimes A
\qquad
\whitecomult: A\to I 
\]
satisfying the following equations, called \emph{associativity} (top),
\emph{coassociativity}  (bottom), \emph{(co)unitality}, and the \emph{Frobenius condition}:
\ctikzfig{fa_axioms}
A Frobenius algebra is \emph{symmetric} when the following equations hold:
\[
\beginpgfgraphicnamed{symmetric_1}
\begin{tikzpicture}[dotpic]
	\begin{pgfonlayer}{nodelayer}
		\node [style=white dot] (0) at (-1.5, -0.25) {};
		\node [style=none] (1) at (-2, 0.25) {};
		\node [style=none] (2) at (-1, 0.25) {};
		\node [style=none] (3) at (-1, 1.25) {};
		\node [style=none] (4) at (-2, 1.25) {};
		\node [style=none] (5) at (1, 1.25) {};
		\node [style=white dot] (6) at (1.5, -0.25) {};
		\node [style=none] (7) at (2, 0.5) {};
		\node [style=none] (8) at (2, 1.25) {};
		\node [style=none] (9) at (1, 0.5) {};
		\node [style=none] (10) at (0, 0.25) {$=$};
		\node [style=white dot] (11) at (1.5, -1) {};
		\node [style=white dot] (12) at (-1.5, -1) {};
	\end{pgfonlayer}
	\begin{pgfonlayer}{edgelayer}
		\draw [in=-90, out=154] (0) to (1.center);
		\draw [style=diredge, in=-90, out=90] (1.center) to (3.center);
		\draw [in=-90, out=30] (0) to (2.center);
		\draw [style=diredge, in=-90, out=90] (2.center) to (4.center);
		\draw [in=-90, out=154] (6) to (9.center);
		\draw [style=diredge, in=-90, out=90, looseness=0.75] (9.center) to (5.center);
		\draw [in=-90, out=30] (6) to (7.center);
		\draw [style=diredge, in=-90, out=90, looseness=0.75] (7.center) to (8.center);
		\draw [style=diredge] (12) to (0);
		\draw [style=diredge] (11) to (6);
	\end{pgfonlayer}
\end{tikzpicture}}
\endpgfgraphicnamed \qquad\qquad %
\beginpgfgraphicnamed{symmetric_2}
\begin{tikzpicture}[dotpic]
	\begin{pgfonlayer}{nodelayer}
		\node [style=white dot] (0) at (-1.5, 0.25) {};
		\node [style=none] (1) at (-2, -1.25) {};
		\node [style=none] (2) at (-2, -0.25) {};
		\node [style=none] (3) at (-1, -0.25) {};
		\node [style=none] (4) at (0, 0) {$=$};
		\node [style=none] (5) at (2, -0.5) {};
		\node [style=none] (6) at (1, -0.5) {};
		\node [style=none] (7) at (2, -1.25) {};
		\node [style=none] (8) at (-1, -1.25) {};
		\node [style=white dot] (9) at (1.5, 0.25) {};
		\node [style=none] (10) at (1, -1.25) {};
		\node [style=white dot] (11) at (1.5, 1) {};
		\node [style=white dot] (12) at (-1.5, 1) {};
	\end{pgfonlayer}
	\begin{pgfonlayer}{edgelayer}
		\draw [style=rdiredge, in=90, out=-154] (0) to (2.center);
		\draw [in=90, out=-90] (2.center) to (8.center);
		\draw [style=rdiredge, in=90, out=-30] (0) to (3.center);
		\draw [in=90, out=-90] (3.center) to (1.center);
		\draw [style=rdiredge, in=90, out=-154] (9) to (6.center);
		\draw [in=90, out=-90, looseness=0.75] (6.center) to (10.center);
		\draw [style=rdiredge, in=90, out=-30] (9) to (5.center);
		\draw [in=90, out=-90, looseness=0.75] (5.center) to (7.center);
		\draw [style=diredge] (0) to (12);
		\draw [style=diredge] (9) to (11);
	\end{pgfonlayer}
\end{tikzpicture}}
\endpgfgraphicnamed
\]
A \emph{dagger Frobenius algebra} is a Frobenius algebra that
additionally satisfies the following equation:
\[
  \whitecomult = \big(\whitemult\big)^\dag \qquad\qquad \whitecounit = \big(\whiteunit\big)^\dag
\]
\end{definition}

Symbolically, we denote the \emph{multiplication of two points} $p, q \colon I\to A$,  that is, 
\[
\whitemult \circ (p \otimes q) \colon I \to A\ ,
\]
as $p \cdot q$.  Also, since the multiplication fixes its unit, and
the dagger fixes the comultiplication given the multiplication, we
will usually represent our algebras as $(A,\whitemult)$. 

\begin{remark}
  In~\cite{QPL_CPstar}, rather than symmetry, the stronger condition of \emph{normalisability} 
  is used. As this condition implies symmetry for dagger Frobenius algebras~\cite[Theorem~2.6]{QPL_CPstar},
  the results in this paper apply unchanged to normalisable Frobenius algebras.
\end{remark}

We write \FHilb for the category of finite-dimensional Hilbert spaces
and linear maps, with the tensor product as the monoidal structure,
and linear adjoint as the dagger.

\begin{theorem}[\cite{VicaryCstar}]\label{thm:dfa-cstar}
  Symmetric dagger Frobenius algebras in \FHilb are in
  1-to-1 correspondence with finite-dimensional C*-algebras.  
  \qed
\end{theorem} 

Recall that \FHilb is a \emph{compact category}~\cite{KellyLaplaza}, that is,  we can coherently pick a \emph{compact structure} on each object as follows.  If ${\mathcal H}$ is a Hilbert space and ${\mathcal H}^*$ is its conjugate space,   the triple:
\[
\left(
{\mathcal H},\
\epsilon_{\mathcal H}\colon \mathbb{C}\to  {\mathcal H}^*\otimes {\mathcal H}:: 1\mapsto \sum_i |i\rangle\otimes|i\rangle,\ 
\eta_{\mathcal H}\colon{\mathcal H}\otimes {\mathcal H}^*\to \mathbb{C}:: |\psi\rangle\otimes|\phi\rangle\mapsto \langle \psi|\phi\rangle
\right)
\]
is a compact structure which can be shown to be independent of the choice of basis--see \cite{CatsII} for more details.  
We depict the maps $\epsilon_{\mathcal H}$ and $\eta_{\mathcal H}$ respectively as:
\[
\beginpgfgraphicnamed{cat_cap_cup1}
\begin{tikzpicture}[dotpic]
	\begin{pgfonlayer}{nodelayer}
		\node [style=none] (0) at (-5.25, 0.25) {};
		\node [style=none] (1) at (-4, 0.25) {};
		\node [style=none] (2) at (3, -0.25) {};
		\node [style=none] (3) at (4.25, -0.25) {};
	\end{pgfonlayer}
	\begin{pgfonlayer}{edgelayer}
		\draw [style=diredge, in=90, out=90, looseness=1.75] (2.center) to (3.center);
		\draw [style=diredge, in=-90, out=-90, looseness=1.75] (0.center) to (1.center);
	\end{pgfonlayer}
\end{tikzpicture}}
\endpgfgraphicnamed
\]
and compactness means that they satisfy:
\[
\beginpgfgraphicnamed{cat_cap_cup2}
\begin{tikzpicture}[dotpic]
	\begin{pgfonlayer}{nodelayer}
		\node [style=none] (0) at (-1.5, 0.75) {};
		\node [style=none] (1) at (0, 0.75) {};
		\node [style=none] (2) at (4, 0.75) {};
		\node [style=none] (3) at (8, 0.75) {};
		\node [style=none] (4) at (-4, 0) {};
		\node [style=none] (5) at (-2.75, 0) {};
		\node [style=none] (6) at (-1.5, 0) {};
		\node [style=none] (7) at (-0.75, 0) {$=$};
		\node [style=none] (8) at (4, 0) {};
		\node [style=none] (9) at (5.25, 0) {};
		\node [style=none] (10) at (6.5, 0) {};
		\node [style=none] (11) at (7.25, 0) {$=$};
		\node [style=none] (12) at (-4, -0.75) {};
		\node [style=none] (13) at (0, -0.75) {};
		\node [style=none] (14) at (6.5, -0.75) {};
		\node [style=none] (15) at (8, -0.75) {};
	\end{pgfonlayer}
	\begin{pgfonlayer}{edgelayer}
		\draw [style=diredge, in=-90, out=-90, looseness=1.75] (8.center) to (9.center);
		\draw [style=diredge, in=-90, out=-90, looseness=1.75] (5.center) to (6.center);
		\draw [style=diredge, in=90, out=90, looseness=1.75] (9.center) to (10.center);
		\draw (6.center) to (0.center);
		\draw (10.center) to (14.center);
		\draw [style=diredge] (3.center) to (15.center);
		\draw (2.center) to (8.center);
		\draw [style=diredge] (13.center) to (1.center);
		\draw [style=diredge, in=90, out=90, looseness=1.75] (4.center) to (5.center);
		\draw (12.center) to (4.center);
	\end{pgfonlayer}
\end{tikzpicture}}
\endpgfgraphicnamed \ .
\]
 
Each symmetric dagger Fobenius algebra also canonically induces a `self-dual' compact structure.  The cups and caps of this compact structure are given by:
\[
\beginpgfgraphicnamed{cap_cup1}
\begin{tikzpicture}[dotpic]
	\begin{pgfonlayer}{nodelayer}
		\node [style={white dot}] (0) at (-5.25, -0.25) {};
		\node [style={white dot}] (1) at (-2.75, -0) {};
		\node [style={white dot}] (2) at (-2.75, -0.75) {};
		\node [style=none] (3) at (-5.75, 0.5) {};
		\node [style=none] (4) at (-4.75, 0.5) {};
		\node [style=none] (5) at (-3.25, 0.75) {};
		\node [style=none] (6) at (-2.25, 0.75) {};
		\node [style=none] (7) at (-4, -0) {$:=$};
		\node [style=none] (8) at (5.75, -0.75) {};
		\node [style={white dot}] (9) at (5.25, 0) {};
		\node [style=none] (10) at (4.75, -0.75) {};
		\node [style=none] (11) at (4, 0) {$:=$};
		\node [style={white dot}] (12) at (2.75, 0.25) {};
		\node [style={white dot}] (13) at (5.25, 0.75) {};
		\node [style=none] (14) at (3.25, -0.5) {};
		\node [style=none] (15) at (2.25, -0.5) {};
	\end{pgfonlayer}
	\begin{pgfonlayer}{edgelayer}
		\draw [style=diredge, in=-90, out=165, looseness=1.25] (0) to (3.center);
		\draw [style=diredge, in=-90, out=15, looseness=1.25] (0) to (4.center);
		\draw [style=diredge] (2) to (1);
		\draw [style=diredge, bend left=15, looseness=1.00] (1) to (5.center);
		\draw [style=diredge, bend right=15, looseness=1.00] (1) to (6.center);
		\draw [style=diredge, in=-165, out=90, looseness=1.25] (15.center) to (12);
		\draw [style=diredge, in=-15, out=90, looseness=1.25] (14.center) to (12);
		\draw [style=diredge] (9) to (13);
		\draw [style=diredge, bend left=15, looseness=1.00] (10.center) to (9);
		\draw [style=diredge, bend right=15, looseness=1.00] (8.center) to (9);
	\end{pgfonlayer}
\end{tikzpicture}}
\endpgfgraphicnamed \ ,
\]
and one easily verifies that it follows from the axioms of a symmetric Frobenius algebra that the required `yanking' conditions hold: 
\[
\beginpgfgraphicnamed{cap_cup2}
\begin{tikzpicture}[dotpic]
	\begin{pgfonlayer}{nodelayer}
		\node [style=none] (0) at (-3.5, 0) {};
		\node [style={white dot}] (1) at (-3, -0.75) {};
		\node [style=none] (2) at (-1.5, 0) {};
		\node [style={white dot}] (3) at (-2, 0.75) {};
		\node [style=none] (4) at (-2.5, 0) {};
		\node [style=none] (5) at (-3.5, 0.75) {};
		\node [style=none] (6) at (-1.5, -0.75) {};
		\node [style=none] (7) at (0, 0.75) {};
		\node [style=none] (8) at (0, -0.75) {};
		\node [style=none] (9) at (-0.75, 0) {$=$};
		\node [style={white dot}] (10) at (3, -0.75) {};
		\node [style=none] (11) at (2.5, 0) {};
		\node [style={white dot}] (12) at (2, 0.75) {};
		\node [style=none] (13) at (3.5, 0) {};
		\node [style=none] (14) at (1.5, 0) {};
		\node [style=none] (15) at (1.5, -0.75) {};
		\node [style=none] (16) at (3.5, 0.75) {};
		\node [style=none] (17) at (0.75, 0) {$=$};
	\end{pgfonlayer}
	\begin{pgfonlayer}{edgelayer}
		\draw [in=-90, out=165, looseness=1.25] (1) to (0.center);
		\draw [style=diredge, in=-165, out=90, looseness=1.25] (4.center) to (3);
		\draw [style=diredge, in=-15, out=90, looseness=1.25] (2.center) to (3);
		\draw [style=diredge] (0.center) to (5.center);
		\draw (6.center) to (2.center);
		\draw [style=diredge] (8.center) to (7.center);
		\draw [in=-90, out=15, looseness=1.25] (1) to (4.center);
		\draw [in=-90, out=15, looseness=1.25] (10) to (13.center);
		\draw [style=diredge, in=-15, out=90, looseness=1.25] (11.center) to (12);
		\draw [style=diredge, in=-165, out=90, looseness=1.25] (14.center) to (12);
		\draw [style=diredge] (13.center) to (16.center);
		\draw (15.center) to (14.center);
		\draw [in=-90, out=165, looseness=1.25] (10) to (11.center);
	\end{pgfonlayer}
\end{tikzpicture}}
\endpgfgraphicnamed \ . 
\]


\section{Abstract projections}

A \emph{projection} of a C*-algebra is a *-idempotent.   In this
section we will recast this definition in light of
Theorem~\ref{thm:dfa-cstar}, that is, we will identify what these
projections are when a C*-algebra is presented as a symmetric dagger
Frobenius algebra in \FHilb, as in~\cite{VicaryCstar}.   

We claim that the projections of a C*-algebra arise as points $p\colon
I\to {\mathcal H}$ satisfying: 
\begin{equation}\label{def:projection}
\beginpgfgraphicnamed{projection}
\begin{tikzpicture}[dotpic]
	\begin{pgfonlayer}{nodelayer}
		\node [style=none] (0) at (-1.5, 0) {$=$};
		\node [style=none] (1) at (-3.5, 1) {};
		\node [style=white dot] (2) at (-3.5, 0) {};
		\node [style=point] (3) at (-2.75, -1) {$p$};
		\node [style=point] (4) at (-4.25, -1) {$p$};
		\node [style=none] (5) at (0, 1) {};
		\node [style=point] (6) at (0, -1) {$p$};
		\node [style=none] (7) at (1.5, 0) {$=$};
		\node [style=copoint] (8) at (3, 0.5) {$p$};
		\node [style=none] (9) at (4.5, 1) {};
		\node [style=white dot] (10) at (3.75, -0.5) {};
		\node [style=white dot] (11) at (3.75, -1.25) {};
	\end{pgfonlayer}
	\begin{pgfonlayer}{edgelayer}
		\draw [style=diredge, bend left] (4) to (2);
		\draw [style=diredge, bend right] (3) to (2);
		\draw [style=diredge] (2) to (1.center);
		\draw [style=diredge] (6) to (5.center);
		\draw [style=diredge] (11) to (10);
		\draw [style=diredge, in=-90, out=150] (10) to (8);
		\draw [style=diredge, in=-90, out=30] (10) to (9.center);
	\end{pgfonlayer}
\end{tikzpicture}}
\endpgfgraphicnamed
\end{equation}
 where the symmetric dagger Frobenius algebra is the one induced by
 Theorem \ref{thm:dfa-cstar}.  Note that the first condition is simply
 idempotence of $\whitemult$-multiplication of points, and the second
 one is \emph{self-conjugateness} with respect to the compact
 structure induced by the symmetric dagger Frobenius algebra.  
Symbolically, we denote this conjugate of $p$ as $p^*$. 

A C*-algebra is realised as a symmetric dagger Frobenius algebra as follows.  Each finite dimensional C*-algebra decomposes as a direct sum of matrix algebras. These can then be represented as \em endomorphism monoids \em ${\rm End}({\mathcal H})$ in \FHilb, which are triples of the following form:
\[
\Bigl(
{\mathcal H}^*\otimes {\mathcal H}\ ,\ 
1_{{\mathcal H}^*}\otimes \eta_{\mathcal H} \otimes 1_{{\mathcal H}}:({\mathcal H}^*\otimes {\mathcal H})\otimes ({\mathcal H}^*\otimes {\mathcal H})\to {\mathcal H}^*\otimes {\mathcal H}\ ,\ 
\epsilon_{\mathcal H}:\mathbb{C}\to {\mathcal H}^*\otimes {\mathcal H}
\Bigr)\ ,
\]
Diagrammatically, for an endomorphism monoid the multiplication and its unit respectively are: 
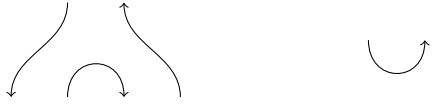
\[
\beginpgfgraphicnamed{pants_alg}
\begin{tikzpicture}[dotpic]
	\begin{pgfonlayer}{nodelayer}
		\node [style=none] (0) at (-1.25, -1.25) {};
		\node [style=none] (1) at (-5.75, -1.25) {};
		\node [style=none] (2) at (-2.75, 1.25) {};
		\node [style=none] (3) at (-4.25, 1.25) {};
		\node [style=none] (4) at (3.75, 0.25) {};
		\node [style=none] (5) at (5.25, 0.25) {};
		\node [style=none] (6) at (-4.25, -1.25) {};
		\node [style=none] (7) at (-2.75, -1.25) {};
	\end{pgfonlayer}
	\begin{pgfonlayer}{edgelayer}
		\draw [style=diredge, in=-90, out=90] (0.center) to (2.center);
		\draw [style=diredge, in=90, out=-90] (3.center) to (1.center);
		\draw [style=diredge, in=-90, out=-90, looseness=2.00] (4.center) to (5.center);
		\draw [style=diredge, in=90, out=90, looseness=2.00] (6.center) to (7.center);
	\end{pgfonlayer}
\end{tikzpicture}}
\endpgfgraphicnamed  
\]
The elements $\rho: \mathbb{C}^n\to \mathbb{C}^n$ of the matrix algebra are then represented by underlying points:
\[
p_\rho\ :=\ \raisebox{-1.2mm}{%
\beginpgfgraphicnamed{CJ}
\begin{tikzpicture}[dotpic]
	\begin{pgfonlayer}{nodelayer}
		\node [style=none] (0) at (-1, 0) {};
		\node [style=square box] (1) at (1, 0.25) {$\rho$};
		\node [style=none] (2) at (-1, 1.25) {};
		\node [style=none] (3) at (1, 1.25) {};
	\end{pgfonlayer}
	\begin{pgfonlayer}{edgelayer}
		\draw [style=diredge, cross, in=-90, out=-90, looseness=1.50] (0.center) to (1);
		\draw (2.center) to (0.center);
		\draw [style=diredge] (1) to (3.center);
	\end{pgfonlayer}
\end{tikzpicture}}
\endpgfgraphicnamed}\ : \ \mathbb{C}\to (\mathbb{C}^n)^*\otimes \mathbb{C}^n
\]
By compactness, each point of type $\mathbb{C}\to
(\mathbb{C}^n)^*\otimes \mathbb{C}^n$ is of this form.  By
Theorem~\ref{thm:dfa-cstar} we know that all symmetric dagger
Frobenius algebras in \FHilb arise in this manner. 

We can now verify the above stated claim on how the projections of a
C*-algebra arise in this representation.  For these points  $p_\rho$
the conditions of equation~\eqref{def:projection} respectively become:   
\[
\beginpgfgraphicnamed{pants_alg_comp}
\begin{tikzpicture}[dotpic]
	\begin{pgfonlayer}{nodelayer}
		\node [style=none] (0) at (-1.5, -1.25) {};
		\node [style=none] (1) at (-5.5, -1.5) {};
		\node [style=none] (2) at (-2.75, 1) {};
		\node [style=none] (3) at (-4.25, 1) {};
		\node [style=none] (4) at (-4.25, -1.25) {};
		\node [style=none] (5) at (-2.75, -1.5) {};
		\node [style=none] (6) at (-5.5, -1.5) {};
		\node [style=square box] (7) at (-4.25, -1.25) {$\rho$};
		\node [style=square box] (8) at (-1.5, -1.25) {$\rho$};
		\node [style=none] (9) at (-2.75, -1.5) {};
		\node [style=none] (10) at (-6.25, -0.5) {$=$};
		\node [style=square box] (11) at (-7.75, -1.25) {$\rho$};
		\node [style=none] (12) at (-9, -1.5) {};
		\node [style=none] (13) at (-7.75, -1.25) {};
		\node [style=none] (14) at (-9, -1.5) {};
		\node [style=none] (15) at (-9, 1) {};
		\node [style=square box] (16) at (-7.75, 0) {$\rho$};
		\node [style=none] (17) at (-7.75, 1) {};
		\node [style=none] (18) at (-0.5, -0.5) {$=$};
		\node [style=none] (19) at (0.5, -1) {};
		\node [style=none] (20) at (1.75, -0.75) {};
		\node [style=none] (21) at (1.75, 0.25) {};
		\node [style=square box] (22) at (1.75, -0.75) {$\rho$};
		\node [style=none] (23) at (0.5, 0.25) {};
		\node [style=none] (24) at (0.5, -1) {};
		\node [style=none] (25) at (3, -0.5) {$=$};
		\node [style=none] (26) at (3.75, 1) {};
		\node [style=none] (27) at (5.25, 0.75) {};
		\node [style=none] (28) at (7.75, 1) {};
		\node [style=none] (29) at (6.5, 0.5) {};
		\node [style=none] (30) at (6.5, -1.5) {};
		\node [style=none] (31) at (5.25, -1.5) {};
		\node [style=none] (32) at (7.75, 1) {};
		\node [style=none] (33) at (3.75, 1) {};
		\node [style=square box] (34) at (5.25, 0.75) {$\rho^\dagger$};
		\node [style=none] (35) at (6.5, -1.5) {};
		\node [style=none] (36) at (5.25, -1.5) {};
		\node [style=none] (37) at (6.5, -1.5) {};
		\node [style=none] (38) at (7.75, 1.75) {};
		\node [style=none] (39) at (6.5, 1.75) {};
		\node [style=none] (40) at (8.5, -0.5) {$=$};
		\node [style=none] (41) at (9.5, 0.25) {};
		\node [style=none] (42) at (10.75, -0.75) {};
		\node [style=none] (43) at (9.5, -1) {};
		\node [style=square box] (44) at (10.75, -0.75) {$\rho^\dagger$};
		\node [style=none] (45) at (9.5, -1) {};
		\node [style=none] (46) at (10.75, 0.25) {};
	\end{pgfonlayer}
	\begin{pgfonlayer}{edgelayer}
		\draw [style=diredge, in=-90, out=90] (0.center) to (2.center);
		\draw [in=90, out=-90] (3.center) to (1.center);
		\draw [in=90, out=90, looseness=2.00] (4.center) to (5.center);
		\draw [style=diredge, cross, in=-90, out=-90, looseness=1.50] (6.center) to (7);
		\draw [style=diredge, cross, in=-90, out=-90, looseness=1.50] (9.center) to (8);
		\draw [in=90, out=-90] (15.center) to (12.center);
		\draw [style=diredge, cross, in=-90, out=-90, looseness=1.50] (14.center) to (11);
		\draw [style=diredge] (13.center) to (16);
		\draw [style=diredge] (16) to (17.center);
		\draw [in=90, out=-90] (23.center) to (24.center);
		\draw [style=diredge, cross, in=-90, out=-90, looseness=1.50] (19.center) to (22);
		\draw [style=diredge] (20.center) to (21.center);
		\draw [in=90, out=-90] (26.center) to (31.center);
		\draw [in=-90, out=90] (30.center) to (32.center);
		\draw [style=diredge, in=-90, out=-90, looseness=2.00] (29.center) to (34);
		\draw [cross, in=90, out=90, looseness=1.50] (27.center) to (33.center);
		\draw [cross, in=-90, out=-90, looseness=1.50] (36.center) to (35.center);
		\draw [style=none] (39.center) to (29.center);
		\draw [style=diredge] (32.center) to (38.center);
		\draw [in=90, out=-90] (41.center) to (43.center);
		\draw [style=diredge, cross, in=-90, out=-90, looseness=1.50] (45.center) to (44);
		\draw [style=diredge] (42.center) to (46.center);
	\end{pgfonlayer}
\end{tikzpicture}}
\endpgfgraphicnamed 
\]
that is, using again compactness, $\rho\circ\rho=\rho=\rho^\dagger$, i.e.~\emph{idempotence} and \emph{self-adjointness}.  

We can now generalise the definition of projection to points $p:I\to A$  to arbitrary symmetric dagger Frobenius algebras $(A,\whitemult)$ in any dagger symmetric monoidal category. 

\begin{definition}
  A \emph{projection} of a symmetric dagger Frobenius algebra 
  $(A,\whitemult)$ in a dagger symmetric monoidal category $\cat{C}$ is
  a morphism $p\colon I\to A$ satisfying equations~\eqref{def:projection}. 
\end{definition}

The next section studies the structure of these generalised
projections of abstract C*-algebras.

Before that, we compare abstract projections to \emph{copyable
points}. These played a key role for commutative abstract C*-algebras,
because they correspond to the elements of an orthonormal basis that
determines the algebra~\cite{CPV}. However, as we will now see, in the
noncommutative case, there simply do not exist enough copyable points
(whereas the projections do have interesting structure, as the next
section shows).
Recall that a point $x \colon I \to A$ is \emph{copyable} when the
following equation is satisfied.
\begin{equation}\label{eq:copyable}
\beginpgfgraphicnamed{copyable}
\begin{tikzpicture}[dotpic]
	\begin{pgfonlayer}{nodelayer}
		\node [style=none] (0) at (-1.5, 0) {$=$};
		\node [style=none] (1) at (-2.75, 1) {};
		\node [style=white dot] (2) at (-3.5, 0) {};
		\node [style=point] (3) at (-3.5, -1) {$x$};
		\node [style=none] (4) at (0, 1) {};
		\node [style=point] (5) at (0, -1) {$x$};
		\node [style=none] (6) at (1, 1) {};
		\node [style=point] (7) at (1, -1) {$x$};
		\node [style=none] (8) at (-4.25, 1) {};
	\end{pgfonlayer}
	\begin{pgfonlayer}{edgelayer}
		\draw [style=diredge] (3) to (2);
		\draw [style=diredge, in=-90, out=15] (2) to (1.center);
		\draw [style=diredge] (5) to (4.center);
		\draw [style=diredge] (7) to (6.center);
		\draw [style=diredge, in=-90, out=165] (2) to (8.center);
	\end{pgfonlayer}
\end{tikzpicture}}
\endpgfgraphicnamed
\end{equation}

\begin{lemma}
  Copyable points of symmetric dagger Frobenius algebras are central.
\end{lemma}
\begin{proof}
  Graphically:
  \[
\beginpgfgraphicnamed{copyable_central}
\begin{tikzpicture}[dotpic]
	\begin{pgfonlayer}{nodelayer}
		\node [style=point] (0) at (-7.75, -1) {$x$};
		\node [style=white dot] (1) at (-7.25, 0) {};
		\node [style=none] (2) at (-7.25, 1) {};
		\node [style=none] (3) at (-6.75, -1) {};
		\node [style=none] (4) at (-6, 0) {$=$};
		\node [style=white dot] (5) at (-4.75, 0) {};
		\node [style=none] (6) at (-5.25, 1) {};
		\node [style=point] (7) at (-4.75, -1) {$x$};
		\node [style=white dot] (8) at (-4, 0.75) {};
		\node [style=none] (9) at (-3.5, -1) {};
		\node [style=none] (10) at (-3.5, 0.25) {};
		\node [style=white dot] (11) at (0, 0.75) {};
		\node [style=point] (12) at (-0.5, 0) {$x$};
		\node [style=none] (13) at (-2.5, 0) {$=$};
		\node [style=point] (14) at (-1.5, 0) {$x$};
		\node [style=none] (15) at (0.5, -1) {};
		\node [style=none] (16) at (-1.5, 1) {};
		\node [style=none] (17) at (0.5, 0.25) {};
		\node [style=none] (18) at (2, -1) {};
		\node [style=none] (19) at (1.25, 0) {$=$};
		\node [style=none] (20) at (4, 1) {};
		\node [style=white dot] (21) at (2.5, 0.75) {};
		\node [style=point] (22) at (4, 0) {$x$};
		\node [style=none] (23) at (2, 0.25) {};
		\node [style=point] (24) at (3, 0) {$x$};
		\node [style=none] (25) at (6, -1) {};
		\node [style=none] (26) at (9.75, 1) {};
		\node [style=white dot] (27) at (7.25, 0) {};
		\node [style=white dot] (28) at (9.75, 0) {};
		\node [style=none] (29) at (8.5, 0) {$=$};
		\node [style=point] (30) at (7.25, -1) {$x$};
		\node [style=white dot] (31) at (6.5, 0.75) {};
		\node [style=none] (32) at (7.75, 1) {};
		\node [style=none] (33) at (5, 0) {$=$};
		\node [style=point] (34) at (10.25, -1) {$x$};
		\node [style=none] (35) at (9.25, -1) {};
		\node [style=none] (36) at (6, 0.25) {};
	\end{pgfonlayer}
	\begin{pgfonlayer}{edgelayer}
		\draw [style=diredge, in=-150, out=90] (0) to (1);
		\draw [style=diredge, in=-33, out=90] (3.center) to (1);
		\draw [style=diredge] (1) to (2.center);
		\draw [style=diredge, in=-90, out=147] (5) to (6.center);
		\draw [style=diredge] (7) to (5);
		\draw [style=diredge, in=-165, out=45] (5) to (8);
		\draw (9.center) to (10.center);
		\draw [style=diredge, in=-15, out=90] (10.center) to (8);
		\draw [style=diredge, in=-165, out=90] (12) to (11);
		\draw (15.center) to (17.center);
		\draw [style=diredge, in=-15, out=90] (17.center) to (11);
		\draw [style=diredge] (14) to (16.center);
		\draw [style=diredge, in=-15, out=90] (24) to (21);
		\draw (18.center) to (23.center);
		\draw [style=diredge, in=-165, out=90] (23.center) to (21);
		\draw [style=diredge] (22) to (20.center);
		\draw [style=diredge, in=-30, out=90] (34) to (28);
		\draw [style=diredge, in=-147, out=90] (35.center) to (28);
		\draw [style=diredge] (28) to (26.center);
		\draw [style=diredge, in=-90, out=33] (27) to (32.center);
		\draw [style=diredge] (30) to (27);
		\draw [style=diredge, in=-15, out=135] (27) to (31);
		\draw (25.center) to (36.center);
		\draw [style=diredge, in=-165, out=90] (36.center) to (31);
	\end{pgfonlayer}
\end{tikzpicture}}
\endpgfgraphicnamed.
  \]
  The middle equation follows from symmetry of $(A,\whitemult)$.
\end{proof}

Let us examine what this implies for the example of $A=(\mathbb{C}^n)^*
\otimes \mathbb{C}^n$ in $\FHilb$ above. Equivalently, we may speak
about $n$-by-$n$ matrices, so that $\whitemult$ becomes actual matrix
multiplication. Because it is well known that the central elements of
matrix algebras are precisely the scalars, any copyable point is
simply a scalar by the previous lemma. But substituting back
into~\eqref{eq:copyable} shows that the only scalar satisfying this
equation is $0$ (unless $n=1$). That is, no noncommutative symmetric dagger 
Frobenius algebra in \FHilb can have nontrivial copyable points. This
explains why we prefer to work with (abstract) projections.

\section{Quantum logics for abstract C*-algebras} 

\begin{definition}
A \emph{zero projection} of $(A,\whitemult)$ is a projection $0 \colon I\to A$  satisfying 
\[
0\cdot p=0
\] 
for all other projections $p\colon I\to A$  of $(A,\whitemult)$. 
\end{definition}

We will assume that an algebra always has a zero projection.

\begin{definition}
An \emph{orthogonality relation} is a binary relation satisfying the following axioms:
\begin{itemize}
\item \emph{symmetry}: $a\perp b \;\Longleftrightarrow\; b\perp a$\ ;
\item \emph{antireflexivity above zero}: $a\perp a \;\Longrightarrow a=0$\ ;
\item \emph{downward closure}: $a\leq a',\; b\leq b',\; a'\perp b' \;\Longrightarrow a\perp b$\,.
\end{itemize}
\end{definition}

\begin{lemma}\label{lemFrob}
We have
\[
\beginpgfgraphicnamed{starproof5}
\begin{tikzpicture}[dotpic]
	\begin{pgfonlayer}{nodelayer}
		\node [style=none] (0) at (-1, 1) {};
		\node [style=none] (1) at (0.5, 1) {};
		\node [style=white dot] (2) at (-0.25, 0) {};
		\node [style=white dot] (3) at (0.5, -0.5) {};
		\node [style=none] (4) at (1.75, 1) {};
	\end{pgfonlayer}
	\begin{pgfonlayer}{edgelayer}
		\draw [style=diredge, in=-90, out=150] (2) to (0.center);
		\draw [style=diredge, in=-90, out=30] (2) to (1.center);
		\draw [style=diredge, bend left=15, looseness=1.25] (3) to (2);
		\draw [style=diredge, bend right] (3) to (4.center);
	\end{pgfonlayer}
\end{tikzpicture}}
\endpgfgraphicnamed = %
\beginpgfgraphicnamed{starproof6}
\begin{tikzpicture}[dotpic]
	\begin{pgfonlayer}{nodelayer}
		\node [style=none] (0) at (-1, 1.25) {};
		\node [style=none] (1) at (0.5, 1.25) {};
		\node [style=white dot] (2) at (1.25, 0) {};
		\node [style=white dot] (3) at (2.25, 0.5) {};
		\node [style=white dot] (4) at (1.25, -1) {};
		\node [style=none] (5) at (2.25, 1.25) {};
	\end{pgfonlayer}
	\begin{pgfonlayer}{edgelayer}
		\draw [style=diredge, in=-150, out=0] (2) to (3);
		\draw [style=diredge, in=-90, out=180, looseness=0.75] (4) to (0.center);
		\draw [style=diredge, in=-90, out=180] (2) to (1.center);
		\draw [style=diredge, in=-30, out=0] (4) to (3);
		\draw [style=diredge] (3) to (5.center);
	\end{pgfonlayer}
\end{tikzpicture}}
\endpgfgraphicnamed 
\]
\end{lemma}
\begin{proof}
First, note the following stardard equation for Frobenius algebras:
\ctikzfig{fa_alt}
Then, the result follows from associativity:
\[ %
\beginpgfgraphicnamed{starproof5all}
\begin{tikzpicture}[dotpic]
	\begin{pgfonlayer}{nodelayer}
		\node [style=none] (0) at (-5.5, 1.25) {};
		\node [style=none] (1) at (-4, 1.25) {};
		\node [style={white dot}] (2) at (-4.75, 0.25) {};
		\node [style={white dot}] (3) at (-4, -0.25) {};
		\node [style=none] (4) at (-2.75, 1.25) {};
		\node [style={white dot}] (5) at (-4, -1.25) {};
		\node [style=none] (6) at (-6.25, -0) {$=$};
		\node [style={white dot}] (7) at (-8.25, -0.5) {};
		\node [style={white dot}] (8) at (-9, -0) {};
		\node [style=none] (9) at (-7, 1) {};
		\node [style=none] (10) at (-9.75, 1) {};
		\node [style=none] (11) at (-8.25, 1) {};
		\node [style={white dot}] (12) at (0, -1.25) {};
		\node [style=none] (13) at (1.5, 1.25) {};
		\node [style=none] (14) at (-2, -0) {$=$};
		\node [style={white dot}] (15) at (0.75, 0.25) {};
		\node [style=none] (16) at (-1.25, 1.25) {};
		\node [style={white dot}] (17) at (0, -0.25) {};
		\node [style=none] (18) at (0, 1.25) {};
		\node [style={white dot}] (19) at (5.25, -0.75) {};
		\node [style={white dot}] (20) at (6.25, 0.75) {};
		\node [style={white dot}] (21) at (5.25, 0.25) {};
		\node [style=none] (22) at (3.25, 1.5) {};
		\node [style=none] (23) at (4.5, 1.5) {};
		\node [style=none] (24) at (6.25, 1.5) {};
		\node [style=none] (25) at (2.25, -0) {$=$};
		\node [style={white dot}] (26) at (5.25, -1.5) {};
		\node [style={white dot}] (27) at (10.5, -0.25) {};
		\node [style={white dot}] (28) at (10.5, -1.25) {};
		\node [style=none] (29) at (11.5, 1) {};
		\node [style=none] (30) at (8.5, 1) {};
		\node [style=none] (31) at (7.5, -0) {$=$};
		\node [style=none] (32) at (9.75, 1) {};
		\node [style={white dot}] (33) at (11.5, 0.25) {};
	\end{pgfonlayer}
	\begin{pgfonlayer}{edgelayer}
		\draw [style=diredge, in=-90, out=150, looseness=1.00] (2) to (0.center);
		\draw [style=diredge, in=-90, out=30, looseness=1.00] (2) to (1.center);
		\draw [style=diredge, bend left=15, looseness=1.25] (3) to (2);
		\draw [style=diredge, bend right, looseness=1.00] (3) to (4.center);
		\draw [style=diredge] (5) to (3);
		\draw [style=diredge, in=-90, out=150, looseness=1.00] (8) to (10.center);
		\draw [style=diredge, in=-90, out=30, looseness=1.00] (8) to (11.center);
		\draw [style=diredge, bend left=15, looseness=1.25] (7) to (8);
		\draw [style=diredge, bend right, looseness=1.00] (7) to (9.center);
		\draw [style=diredge, in=-90, out=30, looseness=1.00] (15) to (13.center);
		\draw [style=diredge, in=-90, out=150, looseness=1.00] (15) to (18.center);
		\draw [style=diredge, bend right=15, looseness=1.25] (17) to (15);
		\draw [style=diredge, bend left, looseness=1.00] (17) to (16.center);
		\draw [style=diredge] (12) to (17);
		\draw [style=diredge, in=-150, out=0, looseness=1.00] (21) to (20);
		\draw [style=diredge, in=-90, out=180, looseness=1.00] (19) to (22.center);
		\draw [style=diredge, in=-90, out=180, looseness=1.00] (21) to (23.center);
		\draw [style=diredge, in=-30, out=0, looseness=1.00] (19) to (20);
		\draw [style=diredge] (20) to (24.center);
		\draw [style=diredge] (26) to (19);
		\draw [style=diredge, in=-150, out=0, looseness=1.00] (27) to (33);
		\draw [style=diredge, in=-90, out=180, looseness=1.00] (28) to (30.center);
		\draw [style=diredge, in=-90, out=180, looseness=1.00] (27) to (32.center);
		\draw [style=diredge, in=-30, out=0, looseness=1.00] (28) to (33);
		\draw [style=diredge] (33) to (29.center);
	\end{pgfonlayer}
\end{tikzpicture}}
\endpgfgraphicnamed. \]
\end{proof} 

\begin{lemma}\label{lemsym}
For projections we have:
\begin{itemize}
\item[{\rm (i)}] $(p\cdot q)^* = q^*\cdot p^*$\ ;
\item[{\rm (ii)}] If $p\cdot q$ is a projection, then $p \cdot q = q
  \cdot p$.
\end{itemize}
\end{lemma}
\begin{proof}
 (i) We have
 \[
 (p\cdot q)^* = \left( %
\beginpgfgraphicnamed{starproof1}
\begin{tikzpicture}[dotpic]
	\begin{pgfonlayer}{nodelayer}
		\node [style=none] (0) at (0, 1.25) {};
		\node [style=white dot] (1) at (0, 0.25) {};
		\node [style=point] (2) at (0.75, -0.75) {$q$};
		\node [style=point] (3) at (-0.75, -0.75) {$p$};
	\end{pgfonlayer}
	\begin{pgfonlayer}{edgelayer}
		\draw [style=diredge, bend left] (3) to (1);
		\draw [style=diredge, bend right] (2) to (1);
		\draw [style=diredge] (1) to (0.center);
	\end{pgfonlayer}
\end{tikzpicture}}
\endpgfgraphicnamed \right)^{\!\!\!*}=
\beginpgfgraphicnamed{starproof2}
\begin{tikzpicture}[dotpic]
	\begin{pgfonlayer}{nodelayer}
		\node [style=copoint] (0) at (-1.25, 0.75) {$p$};
		\node [style=copoint] (1) at (0.25, 0.75) {$q$};
		\node [style=white dot] (2) at (-0.5, -0.25) {};
		\node [style=white dot] (3) at (0.25, -0.75) {};
		\node [style=none] (4) at (1.5, 1) {};
	\end{pgfonlayer}
	\begin{pgfonlayer}{edgelayer}
		\draw [style=diredge, in=-90, out=150] (2) to (0);
		\draw [style=diredge, in=-90, out=30] (2) to (1);
		\draw [style=diredge, bend left=15, looseness=1.25] (3) to (2);
		\draw [style=diredge, bend right] (3) to (4.center);
	\end{pgfonlayer}
\end{tikzpicture}}
\endpgfgraphicnamed  =
\beginpgfgraphicnamed{starproof3}
\begin{tikzpicture}[dotpic]
	\begin{pgfonlayer}{nodelayer}
		\node [style=copoint] (0) at (-1, 1) {$p$};
		\node [style=copoint] (1) at (0.5, 1) {$q$};
		\node [style=white dot] (2) at (1.25, -0.25) {};
		\node [style=white dot] (3) at (2.25, 0.5) {};
		\node [style=white dot] (4) at (1.25, -1.25) {};
		\node [style=none] (5) at (2.25, 1.25) {};
	\end{pgfonlayer}
	\begin{pgfonlayer}{edgelayer}
		\draw [style=diredge, in=-150, out=0] (2) to (3);
		\draw [style=diredge, in=-90, out=180, looseness=0.75] (4) to (0);
		\draw [style=diredge, in=-90, out=180] (2) to (1);
		\draw [style=diredge, in=-30, out=0] (4) to (3);
		\draw [style=diredge] (3) to (5.center);
	\end{pgfonlayer}
\end{tikzpicture}}
\endpgfgraphicnamed  = %
\beginpgfgraphicnamed{starproof4}
\begin{tikzpicture}[dotpic]
	\begin{pgfonlayer}{nodelayer}
		\node [style=copoint] (0) at (0.5, 1) {$p$};
		\node [style=copoint] (1) at (-1, 1) {$q$};
		\node [style=white dot] (2) at (0.25, -0.75) {};
		\node [style=white dot] (3) at (2.25, 0.5) {};
		\node [style=white dot] (4) at (1.5, -0.75) {};
		\node [style=none] (5) at (2.25, 1.25) {};
	\end{pgfonlayer}
	\begin{pgfonlayer}{edgelayer}
		\draw [style=diredge, in=-150, out=0] (2) to (3);
		\draw [style=diredge, in=-90, out=180, looseness=0.75] (4) to (0);
		\draw [style=diredge, in=-90, out=180] (2) to (1);
		\draw [style=diredge, in=-30, out=0] (4) to (3);
		\draw [style=diredge] (3) to (5.center);
	\end{pgfonlayer}
\end{tikzpicture}}
\endpgfgraphicnamed = q^*\cdot p^* \ , 
 \]
 where the middle equation follows from Lemma~\ref{lemFrob}.
(ii) If $p\cdot q=r$ then, by self-conjugateness of projections and
(i), $q\cdot p=q^*\cdot p^*=(p\cdot q)^*=r^*=r=p \cdot q$. 
\end{proof} 

\begin{theorem}\label{thm:proj}
  In a dagger symmetric monoidal category, projections on a symmetric dagger Frobenius algebra 
with a zero projection are  partially ordered and come  
  with an orthogonality relation. 
\end{theorem}
\begin{proof}
  The order is defined as $p\leq q \Longleftrightarrow p\cdot q =p$.  Reflexivity follows by the idempotence of projections.  
  If $p\cdot q= p$ and $q\cdot p= q$ then by Lemma \ref{lemsym}  (ii) we have $p=q$, so the order is anti-symmetric. If $p\cdot q= p$ and $q\cdot r= q$ then $p\cdot r = p\cdot q\cdot r= p\cdot q = p$, so the order is transitive. 
  
  Orthogonality is defined as $p\perp q \Longleftrightarrow p\cdot q =0$.    Symmetry follows by Lemma \ref{lemsym}  (ii) and anti-reflexivity above 0 by idempotence of projections.  If $p\cdot p'=p$, $q\cdot q'=q$ and $p'\cdot q' = 0$ then $p\cdot q= p\cdot p' \cdot q'\cdot q= p\cdot 0 \cdot q= p\cdot 0= 0$ where we twice relied on Lemma \ref{lemsym}  (ii). 
\end{proof}

\begin{remark}
The  zero projection guarantees that the partially ordered set has a bottom element. 
\end{remark}

Given a symmetric dagger Frobenius algebra $(A,\whitemult)$, we will
denote the partial order and orthogonality of the previous theorem as $\Proj(A,\whitemult)$.
The following two examples correspond to the ``pure classical''
and the ``pure quantum'' in the ``concrete'' case of \FHilb. 

\begin{example}
  Commutative dagger special Frobenius algebras $(H,\whitemult)$ in
  \FHilb correspond to orthonormal bases of $H$~\cite{CPV}. For
  $\Proj(H,\whitemult)$, we obtain the atomistic Boolean algebra whose
  atoms are the 1-dimensional projections on the basis vectors.   
\end{example}

\begin{example}
  If $H$ is a finite-dimensional Hilbert space with any chosen compact
  structure on it, then $L(H) = (H^* \otimes H, \pantsalg)$ is a
  symmetric dagger Frobenius algebra in \FHilb. For $\Proj(L(H))$ we
  obtain the usual projection lattice of projections $H \to H$, the paradigmatic
  example in~\cite{BvN}.   
\end{example}

\begin{remark}
  In \cite{QPL_CPstar}, it is shown that algebras of the form $(A^*
  \otimes A, \pantsalg)$ are those that realise Selinger's
  CPM--construction as a fragment of the encompassing CP*--construction.
  The commutative dagger special Frobenius algebras were the
  ones used to underpin abstract categories of stochastic maps in~\cite{CPaqPav}. 
\end{remark}

\begin{proposition}
  Let $(A,\whitemult)$ be any symmetric dagger Frobenius algebra in
  any dagger symmetric monoidal category. For $p, q\in
  \Proj(A,\whitemult)$, the following are equivalent:
  \begin{enumerate}[(a)]
  \item $p$ and $q$ commute;
  \item $p\cdot q\in \Proj(A,\whitemult)$;
  \item $p\cdot q$ is the greatest lower bound of $p$ and $q$ in the
    partial order $\Proj(A,\whitemult)$.
  \end{enumerate}
\end{proposition}
\begin{proof}
  Unfold the definitions of Theorem~\ref{thm:proj}.  
\end{proof}

In general, every commutative monoid of idempotents is a
meet-semilattice with respect to the order $p\leq q \Longleftrightarrow
p\cdot q =p$, and if it is furthermore finite, then it is even a (complete)
lattice.  As shown in~\cite{CPaqPav}, in this case the notion of an idempotent
 can be generalised to arbitrary types $A\to B$. Considered
together for all algebras, this always yields a cartesian bicategory of
relations in the sense of Carboni-Walters~\cite{CarboniWalters}.
The conclusion we draw from the previous proposition is the following:
considering noncommutative algebras obstructs the construction of the 
categorical operation of composition. 

\section{Composing quantum logics}\label{sec:composing}

Given two symmetric dagger Frobenius algebras we can define their
\emph{tensor} as follows.
\[
  (A,\whitemult) \otimes (B, \graymult) := (A \otimes B, \prodmult{white dot}{gray dot})
\]
It is easily seen to inherit the entire algebraic structure.  So we
can define a compositional structure on the corresponding partial orders
with orthogonality as follows.
\[
  \Proj(A,\whitemult) \otimes \Proj(B, \graymult) := \Proj(A \otimes B, \prodmult{white dot}{gray dot})\ .
\]
By a \emph{bi-order map} we mean a function of two variables that 
preserves the order in each argument separately when the other
one is fixed (\emph{cf.}\ bilinearity of the tensor product). 

\begin{theorem}
  The following is a bi-order map.
  \begin{align*}
    - \otimes - \colon \Proj(A,\whitemult) \times \Proj(B, \graymult)
    & \to 
    \Proj(A,\whitemult) \otimes \Proj(B, \graymult) \\ (p, q) & \mapsto p\otimes q
  \end{align*}
 If the monoidal structure moreover preserves zeros, that is, if $0_A$ is a (necessarily unique) zero with respect to $A$ then for all $q:I\to B$ we have that $0_A\otimes q$ is a zero  with respect to $A\otimes B$, then the map $- \otimes -$ also preserves orthogonality in each component.
\end{theorem} 
\begin{proof}
If $p\cdot p' = p$ then: 
\[
(p\otimes q)\cdot (p'\otimes q) = %
\beginpgfgraphicnamed{densproof1}
\begin{tikzpicture}[dotpic]
	\begin{pgfonlayer}{nodelayer}
		\node [style=none] (0) at (0, 1.25) {};
		\node [style=white dot] (1) at (0, 0.25) {};
		\node [style=point] (2) at (1.25, -0.75) {$p'$};
		\node [style=point] (3) at (-1.25, -0.75) {$p$};
		\node [style=gray dot] (4) at (1.25, 0.25) {};
		\node [style=point] (5) at (2.5, -0.75) {$q$};
		\node [style=point] (6) at (0, -0.75) {$q$};
		\node [style=none] (7) at (1.25, 1.25) {};
	\end{pgfonlayer}
	\begin{pgfonlayer}{edgelayer}
		\draw [style=diredge, bend left] (3) to (1);
		\draw [style=diredge, bend right] (2) to (1);
		\draw [style=diredge] (1) to (0.center);
		\draw [style=diredge, bend left] (6) to (4);
		\draw [style=diredge, bend right] (5) to (4);
		\draw [style=diredge] (4) to (7.center);
	\end{pgfonlayer}
\end{tikzpicture}}
\endpgfgraphicnamed = %
\beginpgfgraphicnamed{densproof2}
\begin{tikzpicture}[dotpic]
	\begin{pgfonlayer}{nodelayer}
		\node [style=none] (0) at (-0.5, 1.25) {};
		\node [style=white dot] (1) at (-0.5, 0.25) {};
		\node [style=point] (2) at (0, -0.75) {$p'$};
		\node [style=point] (3) at (-1.25, -0.75) {$p$};
		\node [style=gray dot] (4) at (1.75, 0.25) {};
		\node [style=point] (5) at (2.5, -0.75) {$q$};
		\node [style=point] (6) at (1.25, -0.75) {$q$};
		\node [style=none] (7) at (1.75, 1.25) {};
	\end{pgfonlayer}
	\begin{pgfonlayer}{edgelayer}
		\draw [style=diredge, bend left] (3) to (1);
		\draw [style=diredge, bend right] (2) to (1);
		\draw [style=diredge] (1) to (0.center);
		\draw [style=diredge, bend left] (6) to (4);
		\draw [style=diredge, bend right] (5) to (4);
		\draw [style=diredge] (4) to (7.center);
	\end{pgfonlayer}
\end{tikzpicture}}
\endpgfgraphicnamed= (p\cdot p')\otimes (q\cdot q) = p\otimes q\ .
\]
If $p\cdot p' = 0$ then $(p\otimes q)\cdot (p'\otimes q) = (p\cdot p')\otimes (q\cdot q) = 0_A\otimes q=0_{A\otimes B}$. 
\end{proof}

\begin{remark}
The assumption of the existence of zero projections as well as the assumption of monoidal structure preserving zeros, are both comprehended by the single assumption of the existence of a ``zero scalar'', that is, a morphism $0_I: I\to I$ such that for any other morphisms $f, g: A\to B$ we have that $0_I\otimes f=0_I\otimes g$.  We can then define zero projections  $0_A:=\lambda_A\circ (0_I\otimes 1_A)\circ\lambda^{-1}_A$ where  $\lambda_A:A\simeq I\otimes A$.
\end{remark}

\section{Commutativity versus distributivity}

Having abstracted projection lattices to the setting of arbitrary
dagger symmetric monoidal categories, we can now consider other models
than Hilbert spaces. 

We will be interested in the category $\Rel$ of
sets and relations, where the monoidal structure is taken to be
Cartesian product, and the dagger is given by relational converse.
This setting will provide a counterexample to equation~\eqref{eq:NonComDist}.
Here, symmetric dagger Frobenius algebras were identified by
Pavlovic (in the commutative case) and Heunen--Contreras--Cattaneo (in
general) in~\cite{DuskoRelations} and~\cite{HeunenRelations}, respectively. They
are in 1-to-1 correspondence with small groupoids. As it turns out,
even in the commutative case, groupoids may yield nondistributive projection
lattices. 

\begin{proposition}
  Let $\cat{G}$ be a groupoid, and $(G,\whitemult)$ the corresponding
  symmetric dagger Frobenius algebra in $\Rel$. Elements of
  $\Proj(G,\whitemult)$ are in 1-to-1 correspondence with subgroupoids
  of $\cat{G}$, \textit{i.e.} subcategories of $\cat{G}$ that are
  groupoids themselves.
\end{proposition}
\begin{proof}
  This follows directly from~\cite[Theorem~16]{HeunenRelations}. 
\end{proof}

It immediately follows that in $\Rel$, like in $\FHilb$, the abstract
projection lattice is a complete lattice, even though we are not
dealing with finite sets.

\begin{corollary}
  If $(G,\whitemult)$ is a symmetric dagger Frobenius algebra in
  $\Rel$, then $\Proj(G,\whitemult)$ forms a complete lattice.
\end{corollary}
\begin{proof}
  The collection of subgroupoids is closed under arbitrary intersections.
\end{proof}

In fact, for our counterexample to equation~\eqref{eq:NonComDist}, it
suffices to consider groups (\emph{i.e.}\ single-object groupoids). In
this case abstract projections correspond to subgroups, and it is
known precisely under which conditions the lattice of 
subgroupoids is distributive, thanks to the following classical
theorem due to Ore. A group is \emph{locally cyclic} when
any finite subset of its elements generates a cyclic group.

\begin{theorem}\label{thm:dist-cyclic}
  The lattice of subgroups of a group $G$ is distributive if and only
  if $G$ is abelian and locally cyclic.
\end{theorem}
\begin{proof}
  See~\cite[Theorem~4]{ore:groups}.  
\end{proof}

Perhaps the simplest example of an abelian group that is not locally cyclic is
$\mathbb{Z}_2\times \mathbb{Z}_2$. It has three nontrivial subgroups, namely:
\begin{align*}
  a&:= \mathbb{Z}_2\times\{0\}; \\ 
  b&:= \{(0,0), (1, 1)\}; \\
  c&:= \{0\} \times \mathbb{Z}_2.
\end{align*}
But evidently distributivity breaks down:
$a\wedge (b\vee c)=a \neq 0 = (a\wedge b)\vee (a\wedge c)$.

By Theorem~\ref{thm:dist-cyclic}, we know that the converse (distributive $\implies$ commutative) holds for groups, but what about for arbitrary groupoids. Consider the groupoid with two objects $x, y$ and the only non-identity arrows $f : x \to y$ and $f^{-1} : y \to x$. The lattice of subgroupoids has the following Hasse diagram:
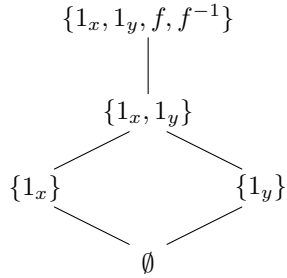
\ctikzfig{counter_ex_lattice}
which is indeed distributive, but $f \circ f^{-1} \neq f^{-1} \circ f$. Thus we have proven the following corollary.

\begin{corollary}
  For symmetric dagger Frobenius algebras $(G,\whitemult)$ in $\Rel$:
  \[\xymatrix{
    (G,\whitemult) \text{ is commutative}\; \ar@<1ex>@{=>}|(.45){\not}[r]
    & \Proj(G,\whitemult) \;\text{ is distributive}. \ar@<1ex>@{=>}|(.55){\not}[l]
  }\]
  \qed
\end{corollary}

Let us finish by remarking on the copyable points in \Rel. As in
\FHilb, they differ from the projections. But unlike in \FHilb, where
there are only trivial copyable points, copyable points in \Rel are
more interesting, for similar reasons as the above corollary.

\begin{lemma}
  If $(G,\whitemult)$ is a symmetric dagger Frobenius algebra in \Rel
  corresponding to a groupoid $\cat{G}$, then its copyable points
  correspond to the connected components of $\cat{G}$.
\end{lemma}
\begin{proof}
  A point $x$ of $G$ in \Rel corresponds to a subset $X \subseteq
  \Mor(\cat{G})$. Copyability now means precisely that
  \[
    X^2 = \{ (g,fg^{-1}) \mid f \in X, g \in \Mor(\cat{G}),
    \dom(f)=\dom(g) \}.
  \]
  Hence if $f \in X$, and $g \in \Mor(\cat{G})$ has $\dom(f)=\dom(g)$,
  then also $g \in X$. Because $\cat{G}$ is a groupoid, this means
  that $X$ is precisely (the set of morphisms of a) connected
  component of $\cat{G}$.
\end{proof}

\section{Further work}

From the point of view of traditional quantum logic, a number of
questions arise, in particular about which categorical structure
yields which order structure:  
\begin{itemize}
\item when is the orthogonality relation an orthocomplementation?
\item when do we obtain an orthoposet? 
\item when do we obtain an orthomodular poset? 
\item when is the partial order a (complete) lattice?
\item when is this lattice Boolean, modular or orthomodular?
\end{itemize}
Conversely, what does the lattice structure say about the category?
An important first step is the characterisation of dagger Frobenius
algebras in more example categories besides \FHilb and \Rel. 

There is a clear intuition of the comultiplication of the algebra being a ``logical broadcasting operation'' in the sense of \cite{CSBayes}. A more general question then arises on the general operational significance of the partial ordering and orthogonality relation constructed in this paper.  

One of the more recent compelling results which emerged from quantum logic is the Faure-Moore-Piron theorem \cite{FMP} on the reconstruction of dynamics from the lattice structure together with the its operational interpretation.  A key ingredient is the reliance on Galois adjoints.  Does this construction have a counterpart within our framework, and its (to still be understood) operational significance?

\bibliographystyle{plain}
\bibliography{qlog}
\end{document}